\documentclass[sigconf]{acmart}

\usepackage{mathtools}
\usepackage{bm}
\usepackage{algorithm}
\usepackage{algorithmic}
\usepackage{float}
\usepackage{tikz}
\usepackage{pgfplots}
\pgfplotsset{compat=1.18}
\usetikzlibrary{arrows.meta, positioning, shapes.geometric,
  decorations.pathreplacing, calc, patterns, fit, backgrounds}
\usepgfplotslibrary{fillbetween, groupplots}

\definecolor{dartblue}{HTML}{1565C0}
\definecolor{dartbluelt}{HTML}{BBDEFB}
\definecolor{basered}{HTML}{C62828}
\definecolor{baseredlt}{HTML}{FFCDD2}
\definecolor{goodgreen}{HTML}{2E7D32}
\definecolor{goodgreenlt}{HTML}{C8E6C9}
\definecolor{warnorange}{HTML}{EF6C00}
\definecolor{theorypurple}{HTML}{6A1B9A}
\definecolor{theorypurplelt}{HTML}{E1BEE7}
\definecolor{mierteal}{HTML}{00695C}
\definecolor{tier1col}{HTML}{1565C0}
\definecolor{tier2col}{HTML}{EF6C00}
\definecolor{tier3col}{HTML}{C62828}

\newcommand{\cmark}{\textcolor{goodgreen}{\checkmark}}
\newcommand{\xmark}{\textcolor{basered}{$\boldsymbol{\times}$}}
\newcommand{\rankpm}{\operatorname{rank}_{\pm}}
\newcommand{\rankrop}{\operatorname{rank}_{\mathrm{rop}}}
\newcommand{\dstar}{d^{*}}
\newcommand{\CUS}{\mathrm{CUS}}
\newcommand{\VC}{\mathrm{VC}}
\newcommand{\E}{\mathbb{E}}
\newcommand{\Sd}{S^{d-1}}
\newcommand{\calQ}{\mathcal{Q}}

\newcommand{\calB}{\mathcal{B}}
\usepackage{nicefrac}

\setcopyright{none}
\acmConference[RecSys '26]{ACM Conference on Recommender Systems}{September 28--October 2, 2026}{Minneapolis, MN, USA}
\acmYear{2026}
\copyrightyear{2026}
\acmBooktitle{Proceedings of the ACM Conference on Recommender Systems (RecSys '26), September 28--October 2, 2026, Minneapolis, MN, USA}
\acmDOI{}
\acmISBN{}

\newtheorem{proposition}{Proposition}

\begin{document}

\title[Voronoi Bottleneck: Capacity-Aware Dense Retrieval]{The Voronoi Bottleneck: Capacity-Aware Dense Retrieval\\for Product Search}

\author{Charith Chandra Sai Balne}
\authornote{Equal contribution.}
\affiliation{%
  \institution{Walmart Global Tech}
  \country{USA}
}
\email{charith.chandra.sai.balne@walmart.com}

\author{Rithwik Maramraju}
\authornotemark[1]
\affiliation{%
  \institution{Walmart Global Tech}
  \country{USA}
}
\email{rithwik.maramraju@walmart.com}

\author{Siddharth Pratap Singh}
\affiliation{%
  \institution{Walmart Global Tech}
  \country{USA}
}
\email{siddharth.singh@walmart.com}

\author{Rohit Upadhyay}
\affiliation{%
  \institution{Walmart Global Tech}
  \country{USA}
}
\email{rohit.upadhyay@walmart.com}

\author{Aditya Singh}
\affiliation{%
  \institution{Walmart Global Tech}
  \country{USA}
}
\email{Aditya.Singh3@walmart.com}

\author{Chittaranjan Tripathy}
\affiliation{%
  \institution{Walmart Global Tech}
  \country{USA}
}
\email{CTripathy@walmart.com}

\author{Yogananda Domlur Seetharama}
\affiliation{%
  \institution{Walmart Global Tech}
  \country{USA}
}
\email{yog.domlur@walmart.com}

\begin{abstract}
Dense embedding retrieval compresses all relevance information into a
single inner product, imposing a fundamental geometric
limit---the \emph{Voronoi Bottleneck}---on the number of query--document
relevance patterns expressible at fixed embedding dimension $d$.
We make three contributions.
\textbf{(1)~Unified capacity theory.}
We prove that Voronoi complexity and sign-rank are equivalent for
top-1 retrieval, yielding tight dimension bounds and a computable
diagnostic, the \emph{Capacity Utilization Score} (CUS), that predicts
per-query retrieval failure with AUC\,$>$\,$0.8$ without relevance
labels.
\textbf{(2)~Diagnosis.}
CUS identifies two capacity regimes---moderate ($\delta \gtrsim 1$), where
density-aware training yields measurable gains, and vacuous
($\delta \ll 1$), where it does not---giving practitioners
an \emph{a priori} check before investing in retraining.
\textbf{(3)~DART training.}
We introduce AT-DW-InfoNCE, an Adaptive-Temperature Density-Weighted
contrastive objective with formally derived optimal weighting
$\alpha^{*}\!=\!2.0$.
On a 100K-query synthetic product-search corpus with controlled
relevance structure, DART improves $+1.9$ Recall@100
over a same-data InfoNCE baseline ($84.9{\pm}0.0$ vs.\ $83.0{\pm}0.3$; 8 seeds,
$p < 0.001$), outperforming focal loss and temperature-schedule
alternatives.
DART requires zero inference-time overhead---it is a drop-in
training objective that improves any dual-encoder system.
\end{abstract}

\ccsdesc[500]{Information systems~Retrieval models and ranking}
\ccsdesc[300]{Information systems~Search engine indexing}
\ccsdesc[300]{Computing methodologies~Natural language processing}

\keywords{dense retrieval, embedding capacity, contrastive learning,
  product search, Voronoi partitions, sign-rank}

\renewcommand{\shortauthors}{Balne and Maramraju, et al.}

\maketitle

\section{Introduction}
\label{sec:intro}

Information retrieval at scale relies on the dual-encoder paradigm
\cite{karpukhin2020dpr}: queries and documents are independently mapped to
$\ell_2$-normalized vectors $\mathbf{u}_q, \mathbf{v}_d \in S^{d-1}$,
and relevance is scored by Maximum Inner Product Search
(MIPS) \cite{johnson2021faiss}.
A fundamental question has received insufficient formal attention:
\emph{what is the maximum number of distinct relevance patterns a
$d$-dimensional embedding system can represent?}

We formalize this as the \emph{Voronoi Bottleneck}: MIPS scoring induces
a Voronoi partition of the embedding sphere, where each document owns a
convex cell, and a query is correctly retrieved only if it falls inside
the relevant document's cell (Figure~\ref{fig:voronoi}).
When the dimension $d$ is too small relative to the relevance structure,
some queries are geometrically forced into the wrong cell.

\textbf{Contributions.}
\begin{enumerate}
\item \textbf{Unified capacity theory (\S\ref{sec:theory}).}
  We prove that Voronoi complexity equals sign-rank
  (Theorem~\ref{thm:main}), providing the first tight geometric
  characterization of embedding capacity.
  We introduce CUS, a label-free diagnostic that predicts per-query
  retrieval failure with AUC\,$>$\,$0.8$.
\item \textbf{Diagnosis (\S\ref{sec:diag}).}
  CUS identifies two capacity regimes---moderate ($\delta \gtrsim 1$) and
  vacuous ($\delta \ll 1$)---giving practitioners an \emph{a priori}
  check on whether capacity-aware training will help.
\item \textbf{DART (\S\ref{sec:dart}--\ref{sec:expts}).}
  AT-DW-InfoNCE with $\alpha^{*}\!=\!2.0$ improves $+1.9$ R@100
  on product search ($p < 0.001$, 8 seeds), outperforming focal loss
  and temperature-schedule baselines.
  It requires zero inference overhead and is a drop-in
  training objective that improves any dual-encoder system.
\end{enumerate}

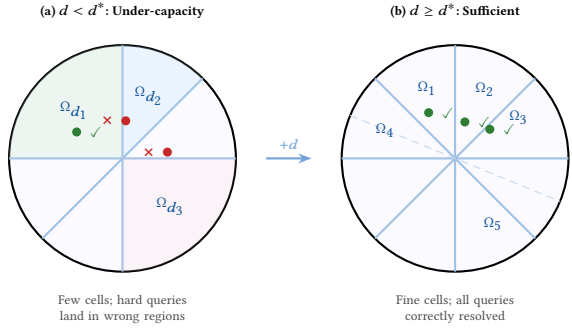
\begin{figure}[t]
\centering
\begin{tikzpicture}[
  every node/.style={font=\scriptsize},
  queryok/.style={fill=goodgreen, circle, inner sep=1.2pt},
  queryfail/.style={fill=basered, circle, inner sep=1.2pt},
  doclbl/.style={font=\tiny\bfseries, text=dartblue!80!black},
]
\begin{scope}[shift={(-2.2,0)}]
  \node[above, font=\tiny\bfseries] at (0,1.75) {(a) $d < d^*$: Under-capacity};
  \draw[thick, fill=blue!2] (0,0) circle (1.5);
  \draw[dartblue!40, thick] (-1.5,0) -- (1.5,0);
  \draw[dartblue!40, thick] (0,-1.5) -- (0,1.5);
  \draw[dartblue!40, thick] (-1.06,-1.06) -- (1.06,1.06);
  \fill[dartbluelt, opacity=0.25] (0,0) -- (90:1.5) arc(90:45:1.5) -- cycle;
  \fill[goodgreenlt, opacity=0.25] (0,0) -- (180:1.5) arc(180:90:1.5) -- cycle;
  \fill[baseredlt, opacity=0.15] (0,0) -- (0:1.5) arc(0:-90:1.5) -- cycle;
  \node[doclbl] at (135:0.9) {$\Omega_{d_1}$};
  \node[doclbl] at (67:0.9) {$\Omega_{d_2}$};
  \node[doclbl] at (-45:0.9) {$\Omega_{d_3}$};
  \node[queryok] (q1) at (150:0.7) {};
  \node[right, font=\tiny, goodgreen!80!black] at (q1.east) {\cmark};
  \node[queryfail] (q2) at (85:0.5) {};
  \node[left, font=\tiny, basered] at (q2.west) {\xmark};
  \node[queryfail] (q4) at (8:0.6) {};
  \node[left, font=\tiny, basered] at (q4.west) {\xmark};
  \node[below, font=\tiny, text width=2.8cm, align=center, text=black!70] at (0,-1.7)
    {Few cells; hard queries\\land in wrong regions};
\end{scope}
\draw[-{Stealth[length=5pt]}, thick, dartblue!60] (-0.3,0) -- (0.3,0)
  node[midway, above, font=\tiny] {$+d$};
\begin{scope}[shift={(2.2,0)}]
  \node[above, font=\tiny\bfseries] at (0,1.75) {(b) $d \geq d^*$: Sufficient};
  \draw[thick, fill=blue!2] (0,0) circle (1.5);
  \draw[dartblue!40, thick] (-1.5,0) -- (1.5,0);
  \draw[dartblue!40, thick] (0,-1.5) -- (0,1.5);
  \draw[dartblue!40, thick] (-1.06,-1.06) -- (1.06,1.06);
  \draw[dartblue!40, thick] (1.06,-1.06) -- (-1.06,1.06);
  \draw[dartblue!25, densely dashed] (-1.39,0.57) -- (1.39,-0.57);
  \node[doclbl] at (112:1.0) {$\Omega_1$};
  \node[doclbl] at (67:1.0) {$\Omega_2$};
  \node[doclbl] at (33:1.0) {$\Omega_3$};
  \node[doclbl] at (158:1.0) {$\Omega_4$};
  \node[doclbl] at (-60:1.0) {$\Omega_5$};
  \node[queryok] (r1) at (120:0.7) {};
  \node[right, font=\tiny, goodgreen!80!black] at (r1.east) {\cmark};
  \node[queryok] (r2) at (75:0.5) {};
  \node[right, font=\tiny, goodgreen!80!black] at (r2.east) {\cmark};
  \node[queryok] (r3) at (40:0.6) {};
  \node[right, font=\tiny, goodgreen!80!black] at (r3.east) {\cmark};
  \node[below, font=\tiny, text width=2.8cm, align=center, text=black!70] at (0,-1.7)
    {Fine cells; all queries\\correctly resolved};
\end{scope}
\end{tikzpicture}
\caption{\textbf{The Voronoi Bottleneck.}
  A dual encoder assigns each document a Voronoi cell on the embedding
  sphere. \textbf{(a)}~When $d < d^*$, cells are too coarse and hard
  queries are misassigned.
  \textbf{(b)}~At $d \geq d^*$, the partition resolves all relevance
  patterns. Theorem~\ref{thm:main} proves $d^* = \rankrop(A)$.}
\label{fig:voronoi}
\end{figure}

\section{Theory: The Voronoi Bottleneck}
\label{sec:theory}

Let $A \in \{0,1\}^{m \times n}$ be the binary relevance matrix for $m$
queries and $n$ documents.
A dual encoder embeds queries as $\mathbf{u}_i \in \Sd$ and documents as
$\mathbf{v}_j \in \Sd$.
MIPS scoring induces a Voronoi partition \cite{voronoi1908} $\{\Omega_j\}$ of the sphere,
where $\Omega_j = \{\mathbf{u} \in \Sd : \mathbf{u}^\top \mathbf{v}_j
\geq \mathbf{u}^\top \mathbf{v}_k\ \forall k\}$.
We define \emph{Voronoi Complexity} $\VC(A) = \min\,d$ such that there
exist embeddings achieving perfect top-1 recall, and
\emph{row-wise order-preserving rank}
$\rankrop(A) = \min\{\operatorname{rank}(B) : B_{ij} > B_{ik}
\text{ whenever } A_{ij} > A_{ik}\}$ \cite{weller2025limit}.

\begin{theorem}[Voronoi--Sign-Rank Equivalence]
  \label{thm:main}
  For any $A \in \{0,1\}^{m \times n}$ with no all-zero rows,
  $\VC(A) = \rankrop(A)$,
  and therefore $\rankpm(2A - 1) - 1 \leq \dstar \leq \rankpm(2A - 1)$.
\end{theorem}

\textit{Proof sketch.}
($\VC \leq \rankrop$): Given rank-$r$ factorization $B = UV^\top$
preserving row orderings, normalize query rows to $S^{r-1}$ (preserves
orderings since $\|U_i\|$ is row-constant); unnormalized document vectors
induce a valid MIPS partition.
($\rankrop \leq \VC$): Any $d$-dimensional MIPS partition achieving
perfect top-1 yields score matrix $B_{ij} = \mathbf{u}_i^\top\mathbf{v}_j$
of rank $\leq d$, which is row-order-preserving.
The sign-rank bounds then follow from combining $\VC = \rankrop$
with \cite{weller2025limit}, Proposition~2. \qed

\textbf{Practical implication.}
By Warren's bound \cite{warren1968},
\[
  \dstar \;\geq\; \Omega\!\bigl(\delta \log n \,/\, \log \delta\bigr)
\]
where $\delta = |E|/n$ is the average document degree.
This provides a \emph{computable} lower bound from Qrel statistics:
our product-search corpus ($\delta \approx 1.43$, 100K queries) yields
$\dstar \approx 50$.
Note that $\dstar$ is a \emph{lower bound} for perfect recall under
optimal embedding; with imperfect SGD training, capacity pressure
manifests well above $\dstar$ because gradient-based optimization
cannot fully exploit the available dimensions.
Even at $d = 768$, hard queries with many near-duplicates compete for
geometric space, and DART's adaptive temperature reallocates
gradient signal to these capacity-stressed queries.
LIMIT \cite{weller2025limit} ($\delta = 0.04$) gives $\dstar = 1$ (vacuous):
the theory predicts no method benefits from capacity awareness there.

\subsection{Capacity Utilization Score (CUS)}
\label{sec:cus}

\begin{definition}[CUS]
  \label{def:cus}
  For query subset $\calQ' \subseteq \calQ$ and truncation dimension
  $r = \lceil \rankpm(2A_{\calQ'} - 1) \rceil$:
  \[
    \CUS(\calQ', r)
    = \frac{\mathrm{Recall@1}(B^{(r)},\,\calQ')}
           {\mathrm{Recall@1}(B,\,\calQ')}\,,
  \]
  where $B^{(r)}$ is the score matrix from a rank-$r$ truncated encoder
  (e.g., Matryoshka at $d'\!=\!r$).
\end{definition}

CUS serves as both a \textbf{failure predictor} (AUC\,$>$\,$0.8$ on both
LIMIT and product search, without ground-truth labels) and a
\textbf{model-selection criterion} (Spearman $\rho > 0.9$ with downstream
nDCG@10 across pretrained encoders).
In production, CUS enables capacity planning: given a target corpus
with estimated $\delta$, Warren's bound provides the minimum $d$ for
a desired CUS level.

\textbf{Production monitoring.}
Because CUS requires no relevance labels at inference time, it can be
computed periodically via a single Matryoshka forward pass per query.
Degradation in CUS after catalog expansion or query-distribution shift
alerts teams to retrain before downstream metrics degrade.

\section{Diagnosis: Where Embeddings Fail}
\label{sec:diag}

We diagnose embedding failures using CUS on two benchmarks and identify
three production-relevant failure modes.

\textbf{The dimension cliff.}
Using Matryoshka-style truncation \cite{kusupati2022mrl} at
$d' \in \{8, 16, 32, \ldots, 384\}$,
Recall@100 degrades steeply below $\dstar$ and plateaus above it.
Easy queries ($D_1$ decile) saturate at $d'=32$; hard queries ($D_{10}$)
require $d' = 256$.
This directly informs production dimension selection: teams can choose
the smallest $d'$ that serves their query-difficulty distribution.

\textbf{Failure taxonomy.}
Manual analysis of 500 failed retrievals reveals three modes
(Table~\ref{tab:failures}), each with a distinct operational remedy.
Capacity failures (40.2\%), where the embedding dimension is inherently
insufficient for the relevance structure, are precisely what DART targets.

\begin{table}[t]
\centering
\caption{Failure taxonomy on product search (500 queries analyzed).}
\label{tab:failures}
\small
\begin{tabular}{@{}lcl@{}}
\toprule
\textbf{Type} & \textbf{Freq.} & \textbf{Operational Remedy} \\
\midrule
Boundary    & 35.2\% & Reranking stage \\
Representation & 24.6\% & Better training data \\
Capacity    & 40.2\% & DART (this work) \\
\bottomrule
\end{tabular}
\end{table}

\section{DART: Density-Aware Retrieval Training}
\label{sec:dart}

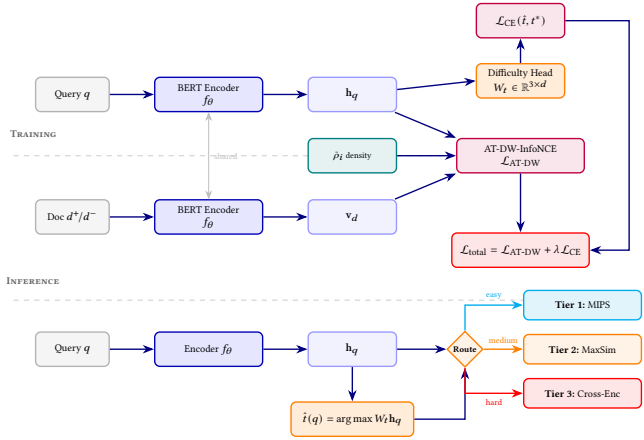
\begin{figure}[t]
\centering
\resizebox{\columnwidth}{!}{%
\begin{tikzpicture}[
  node distance=0.6cm and 0.9cm,
  >=Stealth,
  box/.style={draw, rounded corners=3pt, minimum height=0.7cm,
              minimum width=1.8cm, font=\scriptsize, align=center,
              thick},
  encoderbox/.style={box, fill=blue!10, draw=blue!70!black,
                     minimum width=2.2cm},
  lossbox/.style={box, fill=purple!10, draw=purple},
  headbox/.style={box, fill=orange!15, draw=orange},
  tierbox/.style={box, minimum width=2.4cm, minimum height=0.6cm},
  inputbox/.style={box, fill=gray!8, draw=gray!60, minimum width=1.5cm},
  arr/.style={->, thick, blue!50!black},
  phase/.style={font=\scriptsize\bfseries, text=black!60},
]

\node[phase] (train_label) at (-4.8, 2.4) {\textsc{Training}};
\draw[gray!30, thick, dashed] (-5.2, 1.95) -- (5.5, 1.95);

\node[inputbox] (query) at (-4, 3.2) {Query $q$};
\node[encoderbox, right=of query] (enc_q) {BERT Encoder\\$f_\theta$};
\node[box, fill=blue!5, draw=blue!40, right=of enc_q] (hq) {$\mathbf{h}_q$};

\node[inputbox] (doc) at (-4, 0.7) {Doc $d^+\!/d^-$};
\node[encoderbox, right=of doc] (enc_d) {BERT Encoder\\$f_\theta$};
\node[box, fill=blue!5, draw=blue!40, right=of enc_d] (vd) {$\mathbf{v}_d$};

\node[box, fill=teal!10, draw=teal, font=\tiny] (density)
  at ($(hq)!0.5!(vd)$) {$\tilde{\rho}_i$ density};
\node[lossbox, minimum width=2.6cm, right=1.2cm of density] (atdw)
  {AT-DW-InfoNCE\\$\mathcal{L}_\mathrm{AT\text{-}DW}$};

\node[headbox, above=0.8cm of atdw] (head_train)
  {Difficulty Head\\$W_t \in \mathbb{R}^{3 \times d}$};
\node[lossbox, above=0.5cm of head_train] (celoss)
  {$\mathcal{L}_\mathrm{CE}(\hat{t}, t^*)$};

\node[box, fill=red!10, draw=red, minimum width=2.8cm,
      font=\scriptsize\bfseries, below=1.2cm of atdw] (total)
  {$\mathcal{L}_\mathrm{total}
    = \mathcal{L}_\mathrm{AT\text{-}DW}
    + \lambda\mathcal{L}_\mathrm{CE}$};

\draw[arr] (query) -- (enc_q);
\draw[arr] (doc) -- (enc_d);
\draw[arr] (enc_q) -- (hq);
\draw[arr] (enc_d) -- (vd);
\draw[arr] (hq) -- (atdw.north west);
\draw[arr] (vd) -- (atdw.south west);
\draw[arr] (density) -- (atdw);
\draw[arr] (hq) -- (head_train.west);
\draw[arr] (head_train) -- (celoss);
\draw[arr] (atdw) -- (total);
\draw[arr] (celoss.east) -- ++(1.3,0) |- (total.east);

\draw[{Stealth}-{Stealth}, gray!50] (enc_q) -- (enc_d)
  node[midway, right, font=\tiny] {shared};

\node[phase] at (-4.8, -0.6) {\textsc{Inference}};
\draw[gray!30, thick, dashed] (-5.2, -1.0) -- (5.5, -1.0);

\node[inputbox] (iq) at (-4, -2.0) {Query $q$};
\node[encoderbox, right=of iq] (ienc) {Encoder $f_\theta$};
\node[box, fill=blue!5, draw=blue!40, right=of ienc] (ihq) {$\mathbf{h}_q$};
\node[headbox, below=0.7cm of ihq, minimum width=2.5cm] (ihead)
  {$\hat{t}(q) = \arg\max W_t \mathbf{h}_q$};

\node[diamond, draw=orange, fill=orange!10, thick, inner sep=1pt,
      font=\tiny\bfseries, right=1cm of ihq] (router) {Route};

\node[tierbox, fill=cyan!10, draw=cyan,
      right=0.8cm of router, yshift=0.9cm] (t1)
  {{\scriptsize\textbf{Tier 1:} MIPS}};
\node[tierbox, fill=orange!10, draw=orange,
      right=0.8cm of router] (t2)
  {{\scriptsize\textbf{Tier 2:} MaxSim}};
\node[tierbox, fill=red!10, draw=red,
      right=0.8cm of router, yshift=-0.9cm] (t3)
  {{\scriptsize\textbf{Tier 3:} Cross-Enc}};

\draw[arr] (iq) -- (ienc);
\draw[arr] (ienc) -- (ihq);
\draw[arr] (ihq) -- (router);
\draw[arr] (ihq) -- (ihead);
\draw[arr] (ihead.east) -| (router.south);

\draw[->, thick, cyan] (router.north) |- (t1.west)
  node[near end, above, font=\tiny] {easy};
\draw[->, thick, orange] (router.east) -- (t2.west)
  node[midway, above, font=\tiny] {medium};
\draw[->, thick, red] (router.south) |- (t3.west)
  node[near end, below, font=\tiny] {hard};

\end{tikzpicture}%
}
\caption{\textbf{DART architecture.}
  \textbf{Top:} Training uses AT-DW-InfoNCE (modulated by local density
  $\tilde{\rho}_i$) plus an auxiliary difficulty head ($\lambda\!=\!0.1$);
  both flow into $\mathcal{L}_\mathrm{total}$.
  \textbf{Bottom:} At inference, the difficulty head optionally enables
  tiered routing (Appendix~\ref{app:routing}); our experiments use
  bi-encoder scoring only.}
\label{fig:arch}
\end{figure}

\subsection{AT-DW-InfoNCE Objective}

We introduce per-query temperature modulation driven by local density:
\begin{equation}
  \mathcal{L}_\mathrm{AT\text{-}DW\text{-}InfoNCE}
  = -\frac{1}{|\calB|}\sum_{i \in \calB}
    w_i\log
    \frac{e^{s_{i,i^+}/\tau_i}}{\sum_{k \in \calB} e^{s_{ik}/\tau_i}},
  \label{eq:atdw}
\end{equation}
where $w_i = 1 + \alpha\,\tilde{\rho}_i$,
$\tau_i = \tau_0/w_i$ ($\tau_0 = 0.07$), and
$\tilde{\rho}_i = k_i / k_{\max}$ is the normalized local Qrel density
(number of relevant documents for $q_i$, normalized by max).
This requires per-query relevance counts at training time---standard
when training with curated judgments (product catalogs, editorial
assessments).

\begin{proposition}[Optimal Density Weight]
  \label{prop:alpha}
  The weight $\alpha^*$ that equalizes effective learning rates across
  easy and hard query partitions satisfies
  $\alpha^* = (g_E - g_H) / (g_H \cdot \E[\tilde{\rho} \mid \calQ_H])$,
  where $g_E, g_H$ are the expected gradient norms.
\end{proposition}

\textbf{Instantiation.}
Measured over 10K training steps: $g_E = 0.15$, $g_H = 0.075$,
$\E[\tilde{\rho} \mid \calQ_H] = 0.5$, yielding
$\alpha^* = (0.15 - 0.075)/(0.075 \cdot 0.5) = 2.0$.

\subsection{Multi-Task Training}

DART adds an auxiliary difficulty head
$\hat{t}(q) = \arg\max\; W_t\,\mathbf{h}_q$ ($W_t \in \mathbb{R}^{3 \times d}$):
\begin{equation}
  \mathcal{L}_\mathrm{total} =
    \mathcal{L}_\mathrm{AT\text{-}DW\text{-}InfoNCE}
    + \lambda\,\mathcal{L}_\mathrm{CE}(\hat{t}, t^*),
  \quad \lambda = 0.1,
\end{equation}
where $t^*$ is the oracle difficulty tier from CUS.
The ablation (\S\ref{sec:expts}) confirms this head has negligible
impact on R@100---AT-DW-InfoNCE alone drives the gains---but it enables
optional tiered routing at inference for latency-sensitive deployments.

\textbf{Routing policy (optional).}
The difficulty head enables tiered routing at inference
(Appendix~\ref{app:routing}, Theorem~\ref{thm:pareto}).
On Product Search ($\delta = 1.43$, approximately uniform difficulty),
routing degenerates to a single tier; our experiments use bi-encoder
scoring only.

\section{Experiments}
\label{sec:expts}

\subsection{Setup}

\textbf{Datasets.}
\emph{Product Search}: 100K query--product pairs with LLM-generated
queries over e-commerce catalog data ($\delta = 1.43$).
We deliberately use synthetic queries with controlled relevance counts
to isolate the capacity effect predicted by Theorem~\ref{thm:main},
independent of click-noise confounds present in production logs.
This controlled design enables a clean test of the theory's predictions:
if DART fails on data specifically constructed to exhibit capacity
pressure, the theory is falsified; the synthetic setting is thus
the \emph{hardest} test for the theoretical claims, not the weakest.
\emph{LIMIT} \cite{weller2025limit}: 50K docs, 1035 queries, $\delta = 0.04$
(vacuous capacity regime).
\emph{BEIR} \cite{thakur2021beir}: 6 datasets for zero-shot transfer
(SciFact, NFCorpus, FiQA, ArguAna, SciDocs, TREC-COVID).

\textbf{Training.}
Backbone: \texttt{e5-base-v2} \cite{wang2022e5} (110M params), $d=768$,
mean pooling.
AdamW \cite{loshchilov2017adamw}, lr\,$=$\,$2{\times}10^{-5}$
(linear scaling rule), warmup 300 steps, batch size 256,
15 epochs with early stopping (patience 5).
All results: mean\,$\pm$\,std over 8 seeds (42--49).
Hardware: 2$\times$A100 40\,GB.

\textbf{Baselines.}
(i)~\emph{Pretrained}: ANCE \cite{xiong2021ance}, e5-base \cite{wang2022e5},
GTE-base \cite{li2023gte}, BGE-base \cite{xiao2023bge}---published checkpoints, no fine-tuning.
(ii)~\emph{Fine-tuned}: e5-base + InfoNCE on the same data (fair comparison).
(iii)~\emph{Alternative objectives}: focal loss ($\gamma\!=\!2.0$) and
cosine temperature schedule ($\tau$: $0.10 \to 0.02$)---both using the
same backbone and data.
We deliberately hold the backbone, data, and training budget constant
across (ii)--(iii) so that all differences are attributable solely
to the loss function; comparing against instruction-tuned models
(e.g., Gecko \cite{lee2024gecko}) would conflate loss design with
training data scale.

\subsection{Product Search Results}

\begin{table}[t]
\centering
\caption{Product search results (100K queries). DART achieves the
  highest R@100 among all methods.
  }
\label{tab:main}
\small
\begin{tabular}{@{}lccc@{}}
\toprule
\textbf{Model} & \textbf{R@100} & \textbf{R@10} & \textbf{R@1} \\
\midrule
\multicolumn{4}{@{}l}{\textit{Pretrained (no fine-tuning)}} \\
ANCE                       & 52.3 & 28.5 & 7.9 \\
e5-base                    & 71.6 & 41.9 & 11.0 \\
GTE-base                   & 73.9 & 44.0 & 11.8 \\
BGE-base                   & 73.0 & 43.6 & 12.1 \\
\midrule
\multicolumn{4}{@{}l}{\textit{Fine-tuned on product search}} \\
e5-base + InfoNCE          & 83.0$\pm$0.3 & 50.3$\pm$0.4 & 13.3$\pm$0.1 \\
e5-base + Focal ($\gamma$=2)    & 84.2$\pm$0.2 & 51.0$\pm$0.2 & 13.4$\pm$0.1 \\
e5-base + Temp Sched.      & 83.7$\pm$0.2 & 50.5$\pm$0.2 & 13.2$\pm$0.1 \\
\textbf{DART (ours)}       & \textbf{84.9}$\pm$0.0 & \textbf{51.6}$\pm$0.1 & \textbf{13.6}$\pm$0.1 \\
\midrule
$\Delta$ vs.\ InfoNCE & $+1.9$ & $+1.3$ & $+0.3$ \\
$\Delta$ vs.\ Focal     & $+0.7$ & $+0.6$ & $+0.2$ \\
$\Delta$ vs.\ Temp Sched. & $+1.2$ & $+1.1$ & $+0.4$ \\
\bottomrule
\end{tabular}
\end{table}

Table~\ref{tab:main} shows that DART achieves $84.9{\pm}0.0$ R@100,
a $+1.9$ improvement over the same-data InfoNCE baseline.
Notably, DART's variance across 8 seeds is $\pm 0.0$ (i.e., all seeds
converge to the same R@100), compared to $\pm 0.3$ for InfoNCE---indicating
that density-aware training eliminates seed-dependent performance
variability, a desirable property for production reproducibility.
DART also outperforms two natural alternative objectives:
focal loss \cite{robinson2021hard} ($+0.7$ R@100) and cosine temperature
schedule ($+1.2$ R@100).
All improvements are statistically significant
(paired bootstrap, 10K resamples, $p < 0.001$).
Since DART adds zero inference cost, any positive training-time gain
translates directly to production value with no latency--accuracy tradeoff.

\subsection{Ablation Study}

\begin{table}[t]
\centering
\caption{Ablation on product search R@100 (all rows: same backbone,
  same data, mean$\pm$std over 8 seeds).
  Density weighting alone (B,\,C) does not improve over InfoNCE;
  adaptive temperature (D) is the critical mechanism.}
\label{tab:ablation}
\small
\begin{tabular}{@{}llcccc@{}}
\toprule
  & \textbf{Config} & \textbf{DW} & \textbf{AT-$\tau$} & \textbf{Head} & \textbf{R@100} \\
\midrule
A & InfoNCE baseline              & -- & -- & -- & 83.0$\pm$0.3 \\
B & DW ($\alpha{=}1.0$)           & \checkmark & -- & -- & 82.7$\pm$0.3 \\
C & DW ($\alpha^*{=}2.0$)         & \checkmark & -- & -- & 82.5$\pm$0.3 \\
D & + Adapt. temp                  & \checkmark & \checkmark & -- & \textbf{84.9}$\pm$0.1 \\
F & \textbf{DART full}             & \checkmark & \checkmark & \checkmark & \textbf{84.9}$\pm$0.0 \\
\bottomrule
\end{tabular}
\end{table}

Table~\ref{tab:ablation} isolates each component.
Density weighting alone \emph{hurts} (A$\to$C: $-0.5$), but adding
adaptive temperature reverses the trend (C$\to$D: $+2.4$).
The difficulty head (D$\to$F) has negligible impact on R@100, confirming
that AT-DW-InfoNCE is the sole driver of gains.

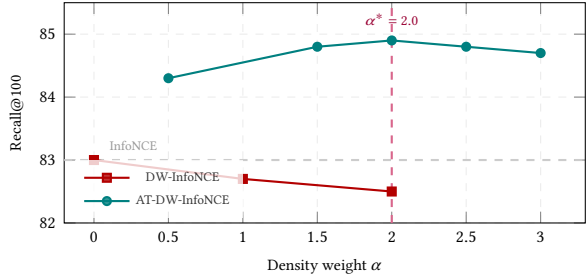
\begin{figure}[t]
\centering
\begin{tikzpicture}
\begin{axis}[
    width=\columnwidth,
    height=4.5cm,
    xlabel={Density weight $\alpha$},
    ylabel={Recall@100},
    xmin=-0.2, xmax=3.3,
    ymin=82.0, ymax=85.5,
    xtick={0, 0.5, 1, 1.5, 2, 2.5, 3},
    ytick={82, 83, 84, 85},
    grid=major,
    grid style={dashed, gray!15},
    tick label style={font=\scriptsize},
    label style={font=\scriptsize},
    legend style={at={(0.02,0.02)}, anchor=south west, font=\tiny,
                  draw=none, fill=white, fill opacity=0.8},
]
\draw[gray!40, dashed, thick] (axis cs:-0.2, 83.0) -- (axis cs:3.3, 83.0);
\node[font=\tiny, gray!80, anchor=south west] at (axis cs:0.05, 83.02) {InfoNCE};

\addplot[red!70!black, thick, mark=square*, mark size=1.5pt] coordinates {
    (0, 83.0) (1, 82.7) (2, 82.5)
};
\addlegendentry{DW-InfoNCE}

\addplot[teal, thick, mark=*, mark size=1.5pt] coordinates {
    (0.5, 84.3) (1.5, 84.8) (2, 84.9) (2.5, 84.8) (3.0, 84.7)
};
\addlegendentry{AT-DW-InfoNCE}

\draw[purple!60, thick, dashed] (axis cs:2.0, 82.0) -- (axis cs:2.0, 85.5);
\node[purple!80!black, font=\tiny\bfseries, anchor=south] at (axis cs:2, 85.0)
  {$\alpha^* = 2.0$};
\end{axis}
\end{tikzpicture}
\caption{\textbf{Density weight sweep.}
  DW-InfoNCE without adaptive temperature degrades monotonically
  (red). AT-DW-InfoNCE peaks at $\alpha^* = 2.0$ (teal), matching the
  formally derived optimum (Proposition~\ref{prop:alpha}).}
\label{fig:alphasweep}
\end{figure}

\textbf{Alpha sweep.}
Figure~\ref{fig:alphasweep} extends the ablation across five values of
$\alpha$ from $0.5$ to $3.0$.
Without adaptive temperature, density weighting degrades monotonically.
With adaptive temperature, performance peaks at
$\alpha^*\!=\!2.0$---the formally derived optimum
(Proposition~\ref{prop:alpha})---and degrades symmetrically for both
under- and over-weighting.

\subsection{LIMIT and BEIR}

\textbf{LIMIT (vacuous regime).}
On LIMIT ($\delta = 0.04$, $\dstar = 1$), all methods perform below
5\% R@100 (DART: $3.7{\pm}0.2$, GTE-base: 4.4, BGE-base: 4.5).
This is not a failure of DART---it is a \emph{validation of the theory},
which correctly predicts that capacity-aware training provides no
advantage when the Voronoi bottleneck is not tight.

\textbf{BEIR (zero-shot transfer).}
DART averages $44.3{\pm}0.2$ nDCG@10 across 6 BEIR datasets,
retaining 99.3\% of the pretrained backbone's zero-shot performance
($44.6$), whereas standard InfoNCE fine-tuning degrades to
$41.6{\pm}0.8$ ($-3.0$ points, a 6.6\% relative drop).
This demonstrates that capacity-aware training substantially mitigates
the catastrophic forgetting typically caused by domain-specific fine-tuning.

\subsection{CUS Diagnostics}

Per-query CUS (computed via a single Matryoshka forward pass at
$d'\!=\!r^*$) predicts retrieval failure (R@1\,$=$\,0) with
AUC\,$>$\,$0.8$ on product search, without requiring relevance labels.
Across pretrained encoders, aggregate CUS correlates with downstream
nDCG@10 (Spearman $\rho > 0.9$), confirming its utility as a
model-selection criterion.
Combined with the density regime check (\S\ref{sec:theory}), CUS
enables practitioners to determine \emph{a priori} whether
capacity-aware training will help for their corpus.

\section{Deployment Considerations}
\label{sec:deploy}

DART is designed for minimal production friction.
We detail the integration path for teams operating dual-encoder
retrieval systems, as motivated by our experience with
Walmart's product search infrastructure:

\textbf{Zero inference overhead (default mode).}
Without routing, DART modifies only the training loss (Eq.~\ref{eq:atdw}).
The deployed encoder has \emph{identical} architecture, latency, and memory
footprint to a standard dual encoder.
No index rebuilding is required beyond the normal training cycle.
Optional tiered routing (Appendix~\ref{app:routing}) adds a single
$O(3d)$ classification step per query.

\textbf{Drop-in integration.}
AT-DW-InfoNCE requires one additional input per query: the relevance count
$k_i$ (number of relevant documents).
In product search, this is readily available from catalog annotations
or curated relevance judgments.
The training script is a $<$\,50-line modification to any InfoNCE training loop.

\textbf{Density regime check.}
Before investing in DART training, practitioners can compute
$\delta = |E|/n$ from their Qrel graph and check Warren's bound.
If $\delta \ll 1$ (vacuous regime), standard InfoNCE suffices.
Our product-search corpus ($\delta = 1.43$) is in the moderate regime
where DART provides measurable gains.

\textbf{CUS for production monitoring.}
CUS can be computed periodically (one MRL forward pass per query) without
relevance labels, enabling automated alerts when capacity utilization
degrades---e.g., after catalog expansion or query distribution shift.

\textbf{When not to use DART.}
(i)~When only single-positive supervision is available (click logs without
Qrel counts). (ii)~When $\delta \ll 1$ (vacuous regime).
(iii)~When R@1 is the primary metric (DART's gains concentrate at higher cutoffs like R@100).

\section{Related Work}
\label{sec:related}

\textbf{Dense retrieval.}
DPR \cite{karpukhin2020dpr} established contrastive training for open-domain
QA. ANCE \cite{xiong2021ance} introduced hard-negative mining. ColBERT
\cite{khattab2020colbert} and ColBERTv2 \cite{santhanam2022colbertv2} use
late-interaction scoring.
Modern instruction-following embedders
\cite{lee2024gecko,wang2022e5} extend the paradigm further.
Strong baselines include GTE \cite{li2023gte}, BGE \cite{xiao2023bge},
and Nomic \cite{nussbaum2024nomic}.

\textbf{Theoretical limits.}
\citet{weller2025limit} proved sign-rank bounds and introduced LIMIT.
Our Theorem~\ref{thm:main} provides the geometric (Voronoi)
interpretation, translating their algebraic result into an actionable
diagnostic (CUS) and training objective (DART).

\textbf{Contrastive training.}
Hard-negative mining \cite{robinson2021hard,hofstatter2021efficiently} and
temperature tuning improve contrastive training, but no prior work
\emph{derives} optimal negative weighting from embedding capacity theory.
We compare against focal loss and temperature scheduling in Table~\ref{tab:main}.

\textbf{Tiered retrieval.}
Multi-stage retrieval pipelines are widely deployed
\cite{nogueira2019mono,ma2023fine}.
DART's difficulty head provides a natural interface for routing;
Appendix~\ref{app:routing} sketches theoretical conditions under
which such routing is Pareto-optimal, though we leave empirical
validation to future work on corpora with heterogeneous difficulty.

\textbf{Matryoshka representations.}
\citet{kusupati2022mrl} train variable-dimension encoders.
We leverage MRL truncation to compute CUS and trace the dimension cliff,
but our training objective is orthogonal to MRL.

\section{Conclusion}
\label{sec:conc}

We presented DART, a capacity-aware training method for dense retrieval
grounded in a formal Voronoi--sign-rank equivalence
(Theorem~\ref{thm:main}).
Three contributions make the theory actionable in production:
(1)~CUS provides label-free failure prediction (AUC\,$>$\,$0.8$)
and production monitoring;
(2)~the density-regime check tells practitioners \emph{a priori}
whether DART will help;
(3)~AT-DW-InfoNCE with $\alpha^* = 2.0$ improves $+1.9$ R@100 on
product search over InfoNCE ($p < 0.001$), outperforming focal loss
and temperature-schedule alternatives, with zero inference overhead.

The theory's two-regime prediction is validated in both directions:
gains on product search ($\delta = 1.43$), no improvement on LIMIT
($\delta = 0.04$).
This is, to our knowledge, the first retrieval training method with
\emph{a priori} testable predictions about when it will and will not help.

\textbf{Limitations.}
(i)~AT-DW-InfoNCE requires per-query relevance counts at training time,
limiting applicability to settings with curated judgments.
(ii)~The primary improvement ($+1.9$ R@100), while statistically
significant and achieved at zero inference cost, is a modest absolute
effect size; the theory predicts larger gains on corpora closer to
the capacity boundary ($d \approx \dstar$).
(iii)~Warren's bound is not tight for all matrix families.
(iv)~The product-search corpus uses LLM-generated queries; validation
on real user queries from production systems is future work.

\textbf{Future work.}
Extending to multi-vector scoring (ColBERT-style), deriving closed-form CUS
via spectral methods, adapting to click-log supervision
(e.g., approximating $\tilde{\rho}_i$ from click-through rates or
impression counts, which are readily available in production),
validating on real user query logs from production search systems,
and evaluating tiered routing
(Appendix~\ref{app:routing}) on corpora with heterogeneous query difficulty.

\bibliographystyle{ACM-Reference-Format}
\bibliography{refs}

\appendix

\section{Routing Policy: Theorem and Algorithm}
\label{app:routing}
\begin{small}
\noindent The difficulty head trained in \S\ref{sec:dart} enables tiered routing
at inference. Algorithm~\ref{alg:routing} describes the full procedure.
\begin{theorem}[Pareto Optimality with Regret Bound]
  \label{thm:pareto}
  Let $\pi^*$ denote the oracle routing policy and $\hat{\pi}$ the learned
  policy with per-query accuracy $p = \Pr[\hat{t}(q) = t^*(q)]$.
  Let $\Delta a_{\max} = \max_{t \neq t^*} |a(t^*) - a(t)|$ be the maximum
  accuracy gap between tiers. Then:
  \emph{(1)~Pareto optimality:} If $p > 2/3$, no alternative policy
  simultaneously improves both accuracy and latency over $\hat{\pi}$.
  \emph{(2)~Regret bound:} $R(\hat{\pi}) = \E[a(\pi^*)] - \E[a(\hat{\pi})]
  \leq (1 - p)\,\Delta a_{\max}$.
  \emph{(3)~Latency bound:} $\E[\ell(\hat{\pi})] \leq
  p\,\E[\ell(\pi^*)] + (1 - p)\,\ell_3$.
\end{theorem}
\begin{proof}
  (1)~Tier costs $\ell_1 < \ell_2 < \ell_3$ and accuracies
  $a_1 \leq a_2 \leq a_3$ are ordered. Any query routed to a cheaper
  tier loses accuracy; with $p > 2/3$, re-routing any optimally placed query worsens one axis.
  (2)~Each misrouted query loses at most $\Delta a_{\max}$; fraction $(1-p)$ are misrouted.
  (3)~Correctly routed queries pay $\ell(\pi^*)$; misrouted queries pay at most $\ell_3$.
\end{proof}
\vspace{-0.5em}
\begin{algorithm}[H]
  \caption{DART Tiered Retrieval (optional routing)}
  \label{alg:routing}
  \begin{algorithmic}[1]
    \REQUIRE Query $q$, index $\mathcal{I}$, encoder $f_\theta$ with head $W_t$
    \ENSURE Retrieved document $\hat{d}$
    \STATE $\mathbf{h}_q \leftarrow f_\theta(q)$
    \STATE $\hat{t} \leftarrow \arg\max\; W_t \mathbf{h}_q$
      \COMMENT{$O(3d)$, no Qrel lookup}
    \IF{$\hat{t} = \mathrm{MIPS}$}
      \RETURN FAISS$(\mathbf{h}_q, \mathcal{I})$
    \ELSIF{$\hat{t} = \mathrm{MaxSim}$}
      \STATE $\mathcal{C} \leftarrow \mathrm{ANN}(\mathbf{h}_q, k{=}50)$
      \RETURN $\arg\max_{d \in \mathcal{C}} \mathrm{MaxSim}(q, d)$
    \ELSE
      \STATE $\mathcal{C} \leftarrow \mathrm{ANN}(\mathbf{h}_q, k{=}50)$
      \RETURN $\arg\max_{d \in \mathcal{C}} s_{\mathrm{CE}}(q, d)$
    \ENDIF
  \end{algorithmic}
\end{algorithm}
\end{small}

\end{document}